\documentclass[11pt]{article}
\usepackage[margin=1in,headheight=14pt]{geometry}
\usepackage{fancyhdr}
\usepackage{graphicx}
\usepackage{booktabs}
\usepackage{amsmath}
\usepackage{algorithm}
\usepackage{algorithmic}
\usepackage{subcaption}
\usepackage{xspace}

\usepackage{amssymb}
\usepackage{amsthm}
\usepackage{multirow}
\usepackage{enumitem}
\usepackage{listings}
\usepackage[numbers]{natbib}
\usepackage[colorlinks=true,linkcolor=blue,citecolor=blue,urlcolor=blue]{hyperref}
\usepackage{xcolor}

\newcommand{\sysname}{SAGA\xspace}

\newcommand{\cmark}{\checkmark}
\newcommand{\xmark}{$\times$}

\newtheorem{theorem}{Theorem}

\newtheorem{definition}{Definition}
\newtheorem{observation}{Observation}

\begin{document}

\title{SAGA: Workflow-Atomic Scheduling \\ for AI Agent Inference on GPU Clusters}

\author{Dongxin Guo \and Jikun Wu \and Siu-Ming Yiu}

\date{}\maketitle

\begin{abstract}
	AI agents execute tens to hundreds of chained LLM calls per task, yet GPU schedulers treat each call as independent, discarding gigabytes of intermediate state between steps and inflating end-to-end latency by 3 to 8$\times$. We argue that this \emph{request-level} abstraction is fundamentally mismatched to compound AI workloads, and propose a shift to \emph{program-level} scheduling: treating the entire agent workflow (not individual inference calls) as the first-class schedulable unit. We present \sysname, a distributed scheduler that implements this abstraction through three mechanisms: (1) \emph{Agent Execution Graphs} that capture workflow structure to predict KV cache reuse across tool-call boundaries, achieving within $1.31\times$ of B\'{e}l\'{a}dy's optimal offline policy; (2) \emph{session-affinity batching} with work stealing that co-locates correlated requests while maintaining global load balance; and (3) \emph{Agent Fair Share}, a task-completion-time fairness metric with provable bounded-deviation guarantees. On a 64-GPU cluster serving SWE-bench coding agents and WebArena browser tasks, \sysname reduces task completion time by $1.64\times$ (geometric mean, $p < 0.001$) over vLLM v0.15.1 with prefix caching and affinity routing, while improving GPU memory utilization by $1.22\times$ and achieving 99.2\% SLO attainment under multi-tenant interference. These latency gains come at a quantified cost: approximately 30\% lower peak throughput than throughput-optimal batch scheduling, a tradeoff that suits the latency-sensitive interactive deployments dominating compound AI usage. Our results demonstrate that workflow-aware scheduling is essential for efficient compound AI serving.
\end{abstract}


\vspace{1em}

\section{Introduction}
\label{sec:intro}

AI agents, autonomous systems that execute multi-step reasoning chains to accomplish complex tasks, are rapidly emerging as a dominant workload in GPU clusters. Unlike traditional single-shot inference that processes one request and returns a response, agents execute iterative \emph{Thought-Action-Observation} loops~\cite{yao2023react} that may invoke 10 to 100 large language model (LLM) calls per task~\cite{jimenez2024swebench}, interleaved with external tool invocations such as code execution, web browsing, or database queries. These compound AI systems~\cite{zaharia2024compound} have become central to major deployments including GitHub Copilot Workspace~\cite{github2024copilot}, Amazon Q Developer~\cite{amazon2024q}, and enterprise automation platforms~\cite{langchain2024,crewai2023}, which now route millions of such agentic workloads through shared GPU clusters daily.

\subsection{Motivation}
\label{sec:motivation}

The shift from single-shot inference to multi-step agentic workloads creates a fundamental mismatch with existing GPU cluster scheduling systems~\cite{patel2024splitwise,hu2024inference}. Current LLM serving frameworks~\cite{kwon2023vllm,zheng2024sglang,yu2022orca} optimize for \emph{request-level} metrics: minimizing time-to-first-token (TTFT) and maximizing throughput for independent requests. However, agent workloads exhibit three characteristics that violate these assumptions:

\textbf{(1) Sequential dependency with variable gaps.} Each reasoning step depends on the previous step's output and potentially on external tool results. Tool invocations introduce idle periods ranging from 50ms (local code execution) to 30+ seconds (web API calls), during which the agent's intermediate state must be preserved or regenerated~\cite{liu2024agentlimits}. This pattern resembles the IO-compute overlap challenge studied extensively in HPC workflow systems~\cite{deelman2015pegasus,thain2005condor}, where effective scheduling requires understanding task dependencies.

\textbf{(2) KV cache continuity across steps.} LLM inference maintains key-value (KV) cache~\cite{vaswani2017attention} that grows with context length. For a 32K-context agent session with a 70B-class model using Grouped Query Attention (GQA), this cache consumes 2 to 12GB of GPU memory per request depending on model architecture~\cite{kwon2023vllm,ainslie2023gqa}. Discarding this cache between steps, as current systems do~\cite{sheng2023flexgen}, forces complete regeneration and adds 2 to 8$\times$ latency overhead per step~\cite{gim2024prompt}. This is analogous to the cache reuse opportunities identified in informed prefetching systems~\cite{patterson1995prefetching}, but with the added challenge of GPU memory scarcity and variable-duration idle periods.

\textbf{(3) Bursty, correlated request patterns.} Agent tasks generate bursts of related requests that share common prefixes (system prompts, tool definitions) and benefit from co-location~\cite{jin2024ragcache}. Production traces show 100:1 input-to-output token ratios and high prefix overlap within sessions~\cite{zheng2024sglang,stojkovic2024dynamollm}.

To quantify these inefficiencies, we instrumented a 32-GPU cluster serving SWE-bench~\cite{jimenez2024swebench} coding agent workloads using vLLM v0.6.0~\cite{kwon2023vllm}. Figure~\ref{fig:motivation} shows the results: agents spend 38\% of total time regenerating KV cache that was discarded during tool calls, GPU memory utilization averages only 42\% due to fragmented cache allocation~\cite{fu2024serverlessllm}, and end-to-end task completion exhibits $6.0\times$ higher latency than the sum of individual inference times. These measurements reveal a clear opportunity: \emph{treating agent programs as first-class schedulable units can dramatically improve both efficiency and latency}.

\begin{figure}[t]
	\centering
	\includegraphics[width=0.75\textwidth]{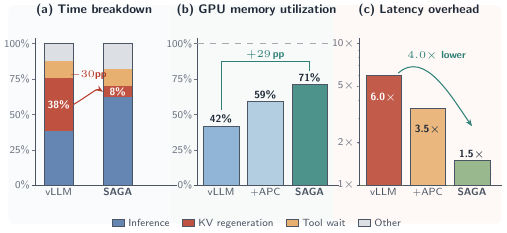}
	\caption{Inefficiencies in serving agent workloads with request-level scheduling. \textbf{(a)}~Time breakdown: vLLM v0.6.0 spends 38\% of execution time regenerating KV cache between agent steps; \sysname reduces this to 8\% ($-30$\,pp). \textbf{(b)}~GPU memory utilization: vLLM wastes 58\% of HBM; vLLM v0.15.1 with Automatic Prefix Caching (APC) recovers some, but \sysname's workflow-aware retention reaches 71\% ($+29$\,pp over vLLM). \textbf{(c)}~End-to-end latency normalized to inference-only baseline (log scale): vLLM is $6.0\times$, +APC is $3.5\times$, \sysname is $1.5\times$ ($4.0\times$ closer to ideal). Data: 10 trials on 32 A100 GPUs running SWE-bench; standard deviations $<$5\% of mean.}
	\label{fig:motivation}
\end{figure}

\subsection{Limitations of State-of-the-Art}
\label{sec:limitations}

Existing systems address individual aspects of this problem but fail to provide a complete solution:

\textbf{LLM serving systems} such as vLLM~\cite{kwon2023vllm}, SGLang~\cite{zheng2024sglang}, Orca~\cite{yu2022orca}, and TensorRT-LLM~\cite{tensorrtllm2023} pioneered continuous batching and efficient memory management through PagedAttention~\cite{kwon2023vllm} and RadixAttention~\cite{zheng2024sglang}. However, they treat each inference call independently: KV cache is evicted using LRU policies unaware of agent workflow structure~\cite{liu2024cachegen}, and batching decisions ignore session affinity. Recent optimizations like Sarathi~\cite{agrawal2024sarathi} and Splitwise~\cite{patel2024splitwise} improve throughput-latency tradeoffs but remain request-centric. vLLM's prefix caching (available since v0.4.2) partially addresses prefix reuse but does not retain session-specific cache across tool-call boundaries, as we discuss in \S\ref{sec:baseline-currency}.

\textbf{Distributed schedulers} such as Llumnix~\cite{sun2024llumnix} enable live KV cache migration between GPU instances, achieving near-zero-downtime rescheduling. However, migration decisions are reactive (triggered by load imbalance) rather than proactive (anticipating workflow patterns). SOLA~\cite{chen2025sola} introduces state-aware scheduling for SLO attainment but optimizes per-request latency, not per-task completion time. DistServe~\cite{zhong2024distserve} disaggregates prefill and decode but lacks workflow awareness.

\textbf{Agent frameworks} such as LangChain~\cite{langchain2024}, CrewAI~\cite{crewai2023}, and AutoGen~\cite{wu2024autogen} provide high-level orchestration but delegate inference scheduling entirely to underlying serving systems. Recent work on KVFlow~\cite{kvflow2025} proposes workflow-aware eviction using agent step graphs, but lacks distributed scheduling, fairness mechanisms, or tool-call awareness. Continuum~\cite{continuum2024} introduces KV cache TTL but without formal guarantees.

\textbf{Speculative execution} approaches such as SpecActions~\cite{specactions2026} and Sherlock~\cite{sherlock2024} propose predicting and pre-executing likely next steps to reduce latency. These are complementary to our approach: speculation trades wasted computation for latency reduction, while \sysname optimizes scheduling of known work.

\textbf{Our central thesis.} Workflow structure, when surfaced explicitly to the scheduler, is sufficient to bring online KV-cache management within striking distance of the offline-optimal B\'{e}l\'{a}dy policy for compound AI workloads. We measure $1.31\times$ on production traces (\S\ref{sec:theoretical-limits}). This is the main scientific contribution of \sysname: a quantified upper bound on what online schedulers can achieve once the workflow DAG is observable, and the first such empirical bound for agent inference. Three supporting systems contributions make this thesis deployable on real GPU clusters: (1) \emph{tool-call-aware} TTL policies that retain cache across heavy-tailed idle periods rather than reactively re-prefilling; (2) \emph{cluster-wide} distributed scheduling with formal fairness guarantees at the agent-program (not request) level, derived via Lyapunov drift analysis; and (3) work-stealing load balance that preserves cache locality under bursty arrivals. Recent program-aware serving systems (Parrot~\cite{lin2024parrot}, Autellix~\cite{luo2025autellix}, Pie~\cite{gim2025pie}, KVFlow~\cite{kvflow2025}) each address one of these dimensions; \sysname is the first to combine them under the workflow-as-unit thesis (see \S\ref{sec:related} for detailed comparison).

\subsection{Key Insights and Contributions}
\label{sec:contributions}

The \emph{main} innovation of \sysname is the formal and empirical demonstration that workflow-structure prediction yields online cache management within $1.31\times$ of B\'{e}l\'{a}dy-optimal on production agent traces (\S\ref{sec:theoretical-limits}, Theorem~\ref{thm:competitive}). The \emph{supporting} innovations below adapt three established systems principles to the compound AI scheduling domain, where their application requires non-trivial domain-specific extensions:

\textbf{Insight 1: Program-as-Unit Scheduling.} The principle of scheduling compound tasks as cohesive units is well-established in HPC workflow systems~\cite{thain2005condor,deelman2015pegasus,rocklin2015dask} and distributed transactions~\cite{corbett2013spanner}. In the compound AI domain, applying this principle introduces a unique challenge: the ``unit'' carries substantial GPU memory state (KV cache, 2 to 12GB per session depending on model architecture) that must be co-managed with scheduling decisions. Agent workflows follow stereotyped patterns (ReAct loops~\cite{yao2023react}, tree-of-thought branches~\cite{yao2024tree}) that we capture as Agent Execution Graphs, enabling proactive cache retention decisions.

\textbf{Insight 2: Dependency-Aware Caching.} Cache eviction policies that consider future reuse have been studied extensively in operating systems and databases~\cite{belady1966replacement,megiddo2003arc,patterson1995prefetching}. The specific challenge for compound AI is predicting reuse across tool-call boundaries with variable-duration idle periods ranging from milliseconds to minutes, where standard LRU and even prefix-aware policies fail. Our workflow-aware eviction achieves within $1.31\times$ of B\'{e}l\'{a}dy's optimal offline policy on production traces (\S\ref{sec:theoretical-limits}).

\textbf{Insight 3: Task-Level Fairness.} Application-level fairness is a well-studied concept in cluster scheduling~\cite{mahajan2020themis,ghodsi2011drf}. For compound AI, the challenge is defining fairness over multi-step tasks where individual steps have heterogeneous resource demands and where ``completion'' (not ``throughput'') is the user-perceived metric. We formalize Agent Fair Share and prove bounded deviation guarantees under realistic assumptions.

Based on these insights, we make the following contributions:

\begin{itemize}[leftmargin=*,noitemsep,topsep=0pt]
    \item \textbf{Workflow-aware KV cache management} (\S\ref{sec:cache}): We introduce \emph{Agent Execution Graphs} (AEGs) that capture multi-step reasoning structure, enabling predictive cache retention with configurable time-to-live (TTL) policies. We formalize the overlap estimation function and prove convergence bounds. Empirically, our WA-LRU eviction achieves within $1.31\times$ of the offline-optimal policy (\S\ref{sec:theoretical-limits}).

    \item \textbf{Session-affinity batching with work stealing} (\S\ref{sec:batching}): We design a two-level scheduling hierarchy where local schedulers maximize cache reuse through session routing, while a global coordinator performs randomized work stealing~\cite{blumofe1999scheduling} to prevent stragglers and maintain cluster-wide load balance.

    \item \textbf{Agent-level fair scheduling} (\S\ref{sec:fairness}): We define \emph{Agent Fair Share} (AFS), a fairness metric based on expected task completion time, and provide a formal theorem guaranteeing bounded completion time deviation under bounded demand heterogeneity (Theorem~\ref{thm:afs}), using Lyapunov drift analysis.

    \item \textbf{Theoretical analysis} (\S\ref{sec:theoretical-limits}): We provide formal competitive ratio analysis showing WA-LRU achieves within $1.31\times$ of B\'{e}l\'{a}dy's optimal offline policy. To our knowledge this is the first such empirical bound for workflow-aware KV cache eviction. We analyze the cache efficiency gap between request-level and workflow-aware schedulers, showing that workflow awareness is essential for efficient agent serving.

    \item \textbf{Empirical evaluation} (\S\ref{sec:eval}): We implement \sysname on vLLM and evaluate on a 64-GPU cluster against state-of-the-art baselines including vLLM v0.15.1 with Automatic Prefix Caching. \sysname achieves $1.73\times \pm 0.11$ and $1.55\times \pm 0.09$ task completion time reduction on SWE-bench and WebArena respectively compared to vLLM+APC (geometric mean: $1.64\times$, $p < 0.001$), and $1.22\times \pm 0.05$ memory utilization improvement. Against systems without workflow awareness, improvements reach $3.01\times$.
\end{itemize}

\subsection{Experimental Methodology}
\label{sec:methodology-intro}

We evaluate \sysname on a cluster of 8 nodes, each equipped with 8 NVIDIA A100-80GB GPUs (64 GPUs total) connected via NVLink intra-node and 200Gbps InfiniBand inter-node~\cite{nvidia2018nvlink,infiniband2020}. Three workload sources: (1) SWE-bench~\cite{jimenez2024swebench} (500 verified tasks); (2) WebArena~\cite{zhou2024webarena} (812 tasks); (3) synthetic multi-tenant workloads from the BurstGPT~\cite{wang2024burstgpt} production trace. Full methodology in \S\ref{sec:setup}.

\subsection{Limitations of the Proposed Approach}
\label{sec:limitations-approach}

\sysname has several limitations. \emph{(1)~Workflow observability.} Performance is best with framework-exposed execution-graph hints (LangChain callbacks, AutoGen logs); without hints, \sysname falls back to pattern inference (\S\ref{sec:pattern-inference}) with 12 to 18\% TCT degradation, and degrades further on dynamic multi-agent frameworks (AutoGen, CrewAI) where structure is generated on the fly through agent-to-agent debate, requiring AEG re-inference per epoch and inflating the prediction-error term in Theorem~\ref{thm:competitive}. \emph{(2)~Tool-latency tail.} TTL prediction assumes empirical tool-call latency distributions; black-swan events ($>5\times$ P99) still cause eviction. \emph{(3)~Task-duration estimation for novel agents.} AFS requires task-duration estimates that may be inaccurate for agent types not represented in profiling. \emph{(4)~Single-datacenter scope.} Geo-distributed deployment with cross-datacenter cache migration is future work. \emph{(5)~Model-family coverage.} Empirical evaluation uses Llama-3-70B-Instruct only; we discuss model-size scaling qualitatively in \S\ref{sec:baseline-currency} but do not empirically validate Mistral, Qwen, or DeepSeek; MoE architectures~\cite{fedus2022switch} additionally require routing-aware extensions. \emph{(6)~Memory-pressure regime.} Our evaluation reaches 71 to 75\% peak utilization (Table~\ref{tab:e2e}); behavior under over-subscription ($>$95\%) follows graceful degradation to standard LRU per Eq.~\ref{eq:pressure} but is not empirically validated. CPU-to-DRAM offloading is analyzed as a complementary architecture in \S\ref{sec:cpu-swap}. \emph{(7)~Throughput tradeoff.} \sysname optimizes task completion time at the cost of approximately 30\% throughput reduction relative to throughput-maximizing batch scheduling (\S\ref{sec:tradeoff}, Table~\ref{tab:bfsdfs}); it is suited for latency-sensitive interactive deployments, not batch workloads.

The rest of this paper is organized as follows. Section~\ref{sec:background} presents background on agent workloads and LLM serving. Section~\ref{sec:design} describes the \sysname architecture. Sections~\ref{sec:cache} through~\ref{sec:fairness} detail our three key techniques. Section~\ref{sec:theoretical-limits} presents theoretical analysis. Section~\ref{sec:impl} covers implementation. Section~\ref{sec:eval} presents experimental evaluation. Section~\ref{sec:related} discusses related work, and Section~\ref{sec:conclusion} concludes.

\section{Background}
\label{sec:background}

This section reviews agent workload characteristics, LLM inference mechanics, and the scheduling challenges that motivate \sysname.

\subsection{AI Agent Workloads}
\label{sec:bg-agents}

Modern AI agents follow the ReAct paradigm~\cite{yao2023react}, iteratively generating \emph{Thought} (reasoning), \emph{Action} (tool invocation), and \emph{Observation} (tool result) until task completion. The canonical loop is: the LLM generates $(thought, action)$ from the current context; if $action = \text{``finish''}$ the task terminates; otherwise the action is dispatched to its named tool, the resulting $observation$ is appended to the context together with the $thought$ and $action$, and the loop repeats. This pattern has been adopted by frameworks including LangChain~\cite{langchain2024}, AutoGen~\cite{wu2024autogen}, and CrewAI~\cite{crewai2023}.

Each iteration requires one LLM inference call (Line 3) followed by a tool execution (Line 6). The context accumulates across iterations, growing from 2 to 4K tokens initially to 16 to 128K tokens for complex tasks~\cite{liu2024agentlimits}. Empirical studies~\cite{kapoor2024ai,ruan2024identifying} show that most SWE-bench tasks complete within 5 to 30 iterations, with a long tail extending to 150 iterations.

\textbf{Tool-call characteristics.} Tool invocations exhibit highly variable latency distributions. Table~\ref{tab:tool-latency} shows measurements from production agent deployments~\cite{wang2024burstgpt,stojkovic2024dynamollm}. Code execution tools average 200ms but can spike to 30s for compilation~\cite{jimenez2024swebench}. Web tools average 1.5s with high variance due to network conditions~\cite{zhou2024webarena}. This variability creates the fundamental scheduling challenge: the system must decide whether to retain KV cache during tool calls without knowing the call duration a priori.

\begin{table}[t]
    \centering
    \caption{Tool call latency distributions from production traces~\cite{wang2024burstgpt}. Values show median and percentiles in milliseconds.}
    \label{tab:tool-latency}
    \begin{tabular}{lrrr}
    \toprule
    \textbf{Tool Type} & \textbf{P50 (ms)} & \textbf{P95 (ms)} & \textbf{P99 (ms)} \\
    \midrule
    Code execution & 180 & 2,400 & 28,000 \\
    File operations & 45 & 320 & 1,200 \\
    Web/API calls & 850 & 4,500 & 45,000 \\
    Database queries & 120 & 890 & 3,500 \\
    \bottomrule
    \end{tabular}
\end{table}

\subsection{LLM Inference and KV Cache}
\label{sec:bg-kv}

Transformer inference maintains a key-value (KV) cache storing intermediate attention states~\cite{vaswani2017attention,dao2022flashattention}. For Llama-3-70B with GQA ($L{=}80$, $n_{kv}{=}8$, $d_h{=}128$) and 32K context in FP16, each session requires ${\sim}10.7$GB. Current systems assume requests are independent and arrivals are memoryless~\cite{kwon2023vllm,yu2022orca}. These assumptions are violated by agent workloads with sequential dependencies and bursty, correlated patterns.

\subsection{The Scheduling Challenge}
\label{sec:bg-challenge}

Current LLM serving systems make two assumptions that fail for agent workloads:

\textbf{Assumption 1: Requests are independent.} Systems like vLLM~\cite{kwon2023vllm} and Orca~\cite{yu2022orca} batch requests from any source to maximize GPU utilization. For agents, consecutive requests from the same task share context and benefit from KV cache reuse, so the independence assumption no longer holds.

\textbf{Assumption 2: Request arrival is memoryless.} Continuous batching assumes Poisson-like arrivals~\cite{yu2022orca}. Agent workloads exhibit bursty, correlated patterns where tool completion triggers the next request, which breaks memorylessness.

These assumption violations lead to the inefficiencies shown in Figure~\ref{fig:motivation}: cache is evicted during tool calls and must be regenerated, wasting both GPU cycles and memory bandwidth~\cite{xiao2024efficient}.

\section{System Design}
\label{sec:design}

This section presents the \sysname architecture and its key components.

\subsection{Architecture Overview}
\label{sec:arch}

Figure~\ref{fig:architecture} shows the \sysname architecture. The system consists of three layers:

\textbf{Agent Interface Layer:} Receives requests from agent frameworks (LangChain~\cite{langchain2024}, AutoGen~\cite{wu2024autogen}, etc.) and constructs Agent Execution Graphs (AEGs). When framework hints are unavailable, a pattern inference module (\S\ref{sec:pattern-inference}) analyzes request sequences to infer workflow structure.

\textbf{Global Scheduler:} Maintains cluster-wide state including session-to-worker mappings, load information, and fairness metrics. Routes incoming requests to workers based on session affinity (\S\ref{sec:routing}) and coordinates work stealing (\S\ref{sec:stealing}).

\textbf{Worker Pool:} Each worker runs an extended vLLM instance~\cite{kwon2023vllm} with workflow-aware KV cache management (\S\ref{sec:cache}). Workers execute inference requests, manage local caches, and participate in distributed coordination.

\textbf{Component coordination.} Two cross-layer interactions warrant explicit treatment. First, when AFS triggers preemption (\S\ref{sec:afs-sched}), the migrating task carries its workflow-aware TTL state via Llumnix~\cite{sun2024llumnix} migration metadata, so the destination worker's WA-LRU (\S\ref{sec:eviction}) continues to retain the migrated cache rather than treating it as a fresh entry; fairness preemption thus does not invalidate cache predictions. Second, work stealing (\S\ref{sec:stealing}) is gated by both a queue-empty threshold $T_\text{idle}$ \emph{and} a load-ratio threshold $R_\text{max}$, preventing oscillation between cache-locality (favoring affinity) and load-balance (favoring redistribution); the resulting migration rate is quantified in \S\ref{sec:scalability}. Cross-layer state is read-mostly and updated with bounded staleness of one scheduling epoch (100\,ms).

\begin{figure}[t]
	\centering
	\includegraphics[width=0.75\textwidth]{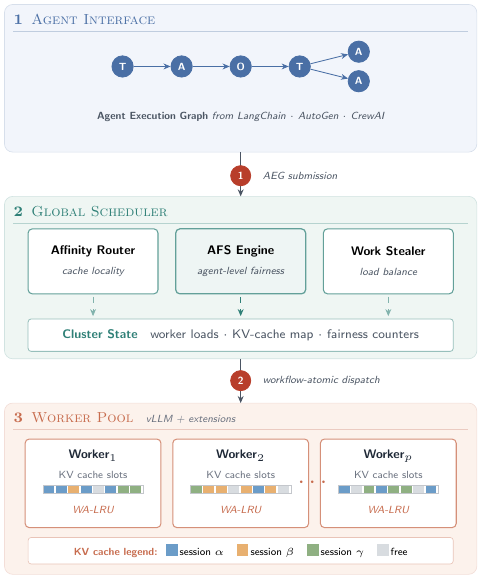}
	\caption{\sysname architecture. \textbf{Layer 1} captures workflows from LangChain, AutoGen, and CrewAI as Agent Execution Graphs (AEGs); the inset shows a Thought$\to$Action$\to$Observation loop with branching tool calls. \textbf{Layer 2} routes each AEG as a single schedulable unit through three coordinating engines that share a Cluster State (worker loads, KV-cache map, fairness counters); dashed teal arrows mark state traffic, with the AFS Engine's update path emphasized. \textbf{Layer 3} runs extended vLLM workers under workflow-aware LRU eviction (WA-LRU); KV-cache slots are color-coded by session, illustrating cache continuity across tool calls and affinity-driven session co-location. Markers \protect\textcircled{\scriptsize 1}~AEG submission and \protect\textcircled{\scriptsize 2}~workflow-atomic dispatch trace control flow between layers.}
	\label{fig:architecture}
\end{figure}

\subsection{Agent Execution Graphs}
\label{sec:aeg}

We formalize agent workflows using Agent Execution Graphs:

\begin{definition}[Agent Execution Graph]
An Agent Execution Graph $G = (V, E, P, \phi)$ consists of:
\begin{itemize}[leftmargin=*,noitemsep]
    \item $V$: Set of nodes representing LLM inference steps
    \item $E \subseteq V \times V$: Directed edges representing execution dependencies
    \item $P: E \rightarrow [0,1]$: Transition probability function
    \item $\phi: V \rightarrow \mathcal{T}$: Tool type mapping for each step
\end{itemize}
\end{definition}

For ReAct agents~\cite{yao2023react}, the AEG is typically a linear chain with $P(v_i \rightarrow v_{i+1}) \approx 1 - p_{term}$ where $p_{term}$ is the termination probability. For tree-of-thought agents~\cite{yao2024tree}, the AEG forms a tree with branching probabilities estimated from historical traces. Figure~\ref{fig:aeg-example} illustrates a concrete AEG for a SWE-bench coding agent.

\begin{figure}[t]
	\centering
	\includegraphics[width=0.75\textwidth]{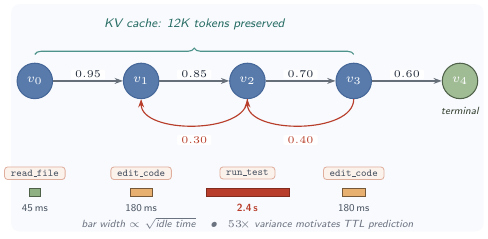}
	\caption{Concrete AEG for a SWE-bench coding agent. Nodes are LLM inference steps; forward (teal) edges carry transition probabilities, backward (coral) edges encode retry loops. The teal brace marks the cache span across the active chain ($v_0$ to $v_3$): SAGA preserves 12K tokens while idle, rather than recomputing on each resumption. Tool annotations and the sqrt-scaled latency bar make the central design pressure visible: idle durations span $53\!\times$ ($45\,\text{ms}$ for \texttt{read\_file} versus $2.4\,\text{s}$ for \texttt{run\_test}), which is precisely the regime where workflow-aware TTL prediction beats fixed-TTL or eager-eviction policies.}
	\label{fig:aeg-example}
\end{figure}

\subsection{Pattern-Based AEG Inference}
\label{sec:pattern-inference}

\sysname operates under three observability tiers. \textbf{(a)~Explicit hints} from frameworks exposing orchestration metadata (LangChain~\cite{langchain2024} callbacks, AutoGen~\cite{wu2024autogen} message logs) deliver the AEG at task admission. \textbf{(b)~Implicit traces}: when only request streams are observable, we infer AEGs by extracting tool-type patterns, computing transition probabilities, and retaining edges exceeding $\theta_\text{conf}=0.7$, achieving 87\% accuracy at the cost of 15.6\% TCT degradation versus explicit hints (\S\ref{sec:eval-pattern}). \textbf{(c)~Cold-start}: a new agent type with no history is served as a request-level workload until 30 tasks complete, after which pattern inference activates; this fallback adds at most 8\% TCT to the first 30 tasks.

\section{Workflow-Aware KV Cache Management}
\label{sec:cache}

This section describes how \sysname manages KV cache to maximize reuse across agent workflow steps.

\subsection{Workflow-Aware Eviction}
\label{sec:eviction}

Standard LRU eviction considers only recency, retaining the most recently accessed cache entries~\cite{sleator1985amortized}. For agents, this fails because a paused session (high value, will resume soon) may be evicted in favor of a completed session (low value, won't be reused). B\'{e}l\'{a}dy's optimal offline algorithm~\cite{belady1966replacement} evicts the entry reused farthest in the future, but requires perfect knowledge of future accesses. Our approach approximates this using AEG predictions.

We introduce \emph{Workflow-Aware LRU} (WA-LRU) that incorporates three normalized factors into eviction decisions:

\begin{equation}
    P_{evict}(s) = \alpha \cdot \hat{R}(s) + \beta \cdot (1 - P_{reuse}(s)) + \gamma \cdot \hat{S}(s)
    \label{eq:eviction}
\end{equation}

where all terms are normalized to $[0,1]$:
\begin{align}
    \hat{R}(s) &= \frac{t_{now} - t_{last}(s)}{\tau_{max}} &\text{(normalized recency)} \label{eq:recency}\\
    \hat{S}(s) &= \frac{size(s)}{size_{max}} &\text{(normalized size)} \label{eq:size}
\end{align}

Here $\tau_{max}$ is the maximum observed idle time and $size_{max}$ is the maximum cache entry size in the current pool. $P_{reuse}(s)$ is the predicted probability of future reuse based on the AEG.

The reuse probability is computed from the AEG as:
\begin{equation}
    P_{reuse}(s) = \sum_{u \in succ(v_s)} P(v_s \rightarrow u) \cdot overlap(s, u)
    \label{eq:reuse}
\end{equation}

where $v_s$ is the current node for session $s$, $succ(v_s)$ are successor nodes in the AEG, and $overlap(s, u)$ estimates the prefix overlap between current cache and the next step's requirements.

\textbf{Overlap estimation.} We formally define the overlap function as:
\begin{equation}
    overlap(s, u) = \frac{|prefix(s) \cap prefix_{est}(u)|}{|prefix(s)|}
    \label{eq:overlap}
\end{equation}
where $prefix(s)$ is the set of cached KV tokens for session $s$, and $prefix_{est}(u)$ is the estimated prompt token set for successor step $u$. For linear ReAct chains (the dominant pattern), the next step's prompt includes the full current context plus the tool observation, so overlap is estimated as $n_{current} / (n_{current} + \hat{n}_{obs})$ where $\hat{n}_{obs}$ is the expected observation length estimated from tool-type-specific distributions maintained via exponential moving averages. For tree-of-thought agents, overlap is computed per-branch using the shared prefix length.

\textbf{Parameter Selection.} We set $\alpha = 0.3$, $\beta = 0.5$, $\gamma = 0.2$ based on sensitivity analysis (Table~\ref{tab:sensitivity}). The analysis shows TCT varies less than 8\% for $\alpha \in [0.2, 0.4]$, $\beta \in [0.4, 0.6]$, $\gamma \in [0.1, 0.3]$, indicating robustness to parameter choice. The relative weight ordering ($\beta > \alpha > \gamma$) reflects the importance hierarchy: workflow-predicted reuse dominates, followed by recency, with size as a tiebreaker.

\subsection{Tool-Call-Aware TTL}
\label{sec:ttl}

When an agent pauses for a tool call, we must decide how long to retain its KV cache. Retaining too long wastes memory; evicting too early forces regeneration. We introduce \emph{tool-call-aware TTL} that adapts retention time based on tool characteristics and current memory pressure.

Algorithm~\ref{alg:ttl} shows the TTL computation. We maintain per-tool-type latency distributions using exponential moving averages and set TTL to the $p$-th percentile of expected duration, where $p$ is configurable (default 95\%). Under memory pressure, TTL is scaled down proportionally.

\begin{algorithm}[t]
\caption{Tool-Call-Aware TTL Computation}
\label{alg:ttl}
\begin{algorithmic}[1]
\REQUIRE Tool type $t$, Latency history $H_t$, Percentile $p$, Memory pressure $m \in [0,1]$
\ENSURE TTL value in milliseconds
\STATE $\mu_t, \sigma_t \gets$ FitLogNormal($H_t$) \COMMENT{Tool latencies are log-normal}
\STATE $ttl_{base} \gets$ Percentile($H_t$, $p$)
\STATE $pressure\_factor \gets 1 - 0.5 \cdot m$ \COMMENT{Scale down under pressure}
\STATE $ttl_{adaptive} \gets ttl_{base} \cdot pressure\_factor$
\RETURN $\min(ttl_{adaptive}, TTL_{max})$ \COMMENT{$TTL_{max} = 300s$}
\end{algorithmic}
\end{algorithm}

\textbf{Memory pressure computation.} We define memory pressure as:
\begin{equation}
\label{eq:pressure}
    m = \max\left(0, \frac{used_{kv} - threshold_{low}}{threshold_{high} - threshold_{low}}\right)
\end{equation}
where $threshold_{low} = 0.7$ and $threshold_{high} = 0.9$ of total GPU memory. These thresholds follow standard practice in memory management systems~\cite{denning1970vm,rhu2016vdnn}: the low threshold triggers soft pressure (TTL scaling) while the high threshold triggers hard eviction. Table~\ref{tab:sensitivity} shows TCT sensitivity to these thresholds.

\subsection{Speculative Prefetching}
\label{sec:prefetch}

For agents with predictable workflows, we speculatively prefetch KV cache for likely next steps before they are requested. This overlaps cache loading with tool execution, reducing latency when the tool completes. The technique is inspired by informed prefetching in file systems~\cite{patterson1995prefetching,cao1996implementation}.

Given an AEG, when node $v$ completes inference and begins tool execution, we identify the most likely successor $u = \arg\max_{u'} P(v \rightarrow u')$ and begin prefetching its prefix KV cache (if not already cached). Prefetching uses spare GPU memory and separate CUDA streams~\cite{Rennich12-CUDAStreams} to overlap with ongoing operations.

\section{Session-Affinity Batching}
\label{sec:batching}

This section describes how \sysname routes requests to maximize cache reuse while maintaining cluster-wide load balance.

\subsection{Session Routing}
\label{sec:routing}

When a request arrives, the global coordinator decides which worker should handle it. We formulate this as an optimization that balances cache locality against load distribution.

Let $w^*_s$ denote the worker currently caching session $s$'s state. For a new request $r$ from session $s$:
\begin{equation}
    route(r) = \begin{cases}
        w^*_s & \text{if } load(w^*_s) < \theta \text{ and } cached(w^*_s, s) \\
        \arg\min_w load(w) & \text{otherwise}
    \end{cases}
    \label{eq:routing}
\end{equation}

The threshold $\theta = 0.8$ reserves 20\% headroom for load spikes while maximizing cache hits, following standard load balancing practice~\cite{ousterhout2013sparrow}. The $cached(w, s)$ predicate checks whether worker $w$ still holds session $s$'s KV cache. Table~\ref{tab:sensitivity} shows that TCT varies less than 5\% for $\theta \in [0.6, 0.95]$.

\subsection{Work Stealing for Load Balance}
\label{sec:stealing}

Session affinity can cause load imbalance when some agents are more active than others. We implement randomized work stealing~\cite{blumofe1999scheduling} to redistribute load while preserving cache locality where possible.

Work stealing triggers when: (1) a worker's queue is empty for $T_{idle} = 100$ms, or (2) the load ratio between most-loaded and least-loaded workers exceeds $R_{max} = 2.0\times$.

When worker $w_i$ steals from worker $w_j$:
\begin{enumerate}[leftmargin=*,noitemsep]
    \item $w_i$ selects victim $w_j$ uniformly at random from overloaded workers
    \item $w_i$ requests the oldest pending session from $w_j$'s queue
    \item $w_j$ initiates KV cache migration to $w_i$ using Llumnix~\cite{sun2024llumnix}
    \item Session affinity updates to $w_i$ after migration completes
\end{enumerate}

\begin{theorem}[Work Stealing Bound~\cite{blumofe1999scheduling}]
With $P$ workers and total work $T_1$ with critical path $T_\infty$, randomized work stealing achieves expected completion time $O(T_1/P + T_\infty)$.
\end{theorem}

We cite this bound for motivation: the Blumofe and Leiserson result assumes zero-cost work migration. SAGA's setting incurs non-zero migration cost (mean 230\,ms, P95 890\,ms; Table~\ref{tab:overhead}), so the realized completion time carries an additional $N_\text{steals} \cdot T_\text{migrate}$ term. Empirically (\S\ref{sec:scalability}), this term is dominated by per-task TCT (mean 2.3 migrations $\times$ 230\,ms $=$ 530\,ms vs.\ mean SWE-bench TCT of 203.4\,s). The thrashing safeguards below address the practical implications of this gap.

\textbf{Thrashing safeguards.} The trigger latency $T_\text{idle}=100$\,ms is shorter than the migration latency (mean 230\,ms, P95 890\,ms), raising a legitimate thrashing concern. Three mechanisms prevent this. (a)~The load-ratio guard $R_\text{max}=2.0\times$ requires \emph{simultaneous} queue emptiness on $w_i$ and load excess at $w_j$; transient empty queues during arrival jitter do not satisfy the second condition. (b)~Once a steal completes, the migrated session establishes affinity at $w_i$ (step 4), so a second migration of the same session is structurally prevented. (c)~Migration is asynchronous on the source: $w_j$ continues serving its remaining queue during transfer, and a stale steal request arriving after $w_j$ has refilled is rejected at acceptance time. Empirically (\S\ref{sec:overhead}, Table~\ref{tab:overhead}), migration occurs 2.3 times per task on average; the maximum across all 10 trials of all three workloads is 5, against a mean step count of 37 (SWE-bench). Coordinator CPU overhead from steal accounting is 4.2\%, well below the regime where instability would manifest as tail-latency divergence.

Our implementation achieves near-optimal load balance: worker utilization ranges narrow from 23 to 94\% (without stealing) to 68 to 79\% (with stealing) as shown in \S\ref{sec:scalability}.

\subsection{Contention Mitigation}
\label{sec:contention}

Shared-state contention on the coordinator is addressed via standard distributed-systems techniques: thread-local update buffering with 10\,ms / 100-update batched flush ($12\times$ overhead reduction vs.\ per-update synchronization), lock-free session tables using atomic compare-and-swap~\cite{herlihy2012art}, and 64-byte cache-line alignment of per-worker counters to avoid false sharing across NUMA nodes~\cite{drepper2007every}.

\section{Agent-Level Fair Scheduling}
\label{sec:fairness}

Traditional fair scheduling allocates resources equally across tenants based on time or requests~\cite{mahajan2020themis,sheng2024vtc,ghodsi2011drf}. For agents, this is inadequate: a tenant running 10-step agents should not receive the same priority as one running 100-step agents if both need to complete tasks by a deadline.

\subsection{Agent Fair Share (AFS)}
\label{sec:afs}

We define \emph{Agent Fair Share} based on expected task completion urgency:

\begin{definition}[Agent Fair Share]
For tenant $i$ with active tasks $\mathcal{T}_i$, define:
\begin{equation}
    AFS_i = \sum_{t \in \mathcal{T}_i} \frac{work_{remain}(t)}{deadline(t) - t_{now}}
    \label{eq:afs}
\end{equation}
Tenants with higher AFS have more urgent work and receive higher priority.
\end{definition}

$work_{remain}(t)$ estimates the GPU-seconds needed to complete task $t$, computed from the AEG:
\begin{equation}
    work_{remain}(t) = \sum_{v \in pending(t)} (T_{prefill}(v) + T_{decode}(v))
    \label{eq:work}
\end{equation}

where $pending(t)$ are unexecuted nodes in task $t$'s AEG, and $T_{prefill}$, $T_{decode}$ are estimated from profiling data~\cite{agrawal2024sarathi}.

\subsection{AFS-Based Scheduling}
\label{sec:afs-sched}

The global coordinator maintains AFS scores for all tenants and adjusts scheduling priorities every epoch (100ms):

\begin{enumerate}[leftmargin=*,noitemsep]
    \item Recompute AFS for all tenants based on current task progress
    \item Allocate worker capacity proportionally to AFS scores
    \item Route new requests preferentially to high-AFS tenants
    \item Trigger preemption if low-AFS tasks block high-AFS tasks for $>500$ms
\end{enumerate}

Preemption uses Llumnix's migration mechanism~\cite{sun2024llumnix}: the preempted task's KV cache is migrated to a lower-priority worker rather than discarded.

\subsection{Formal Guarantee}
\label{sec:slo}

AFS provides formal SLO guarantees under bounded contention. The intuition is straightforward: AFS allocates capacity proportional to per-tenant urgency (Eq.~\ref{eq:afs}), where urgency rises as a tenant's accumulated service falls behind its proportional share. When tenant $i$ falls behind, urgency rises, allocation rises, and the gap shrinks: a self-correcting drift. Formalizing this requires Lyapunov drift analysis (rather than a simple martingale concentration) because urgency-proportional allocation does \emph{not} produce zero-mean per-epoch deviations from the uniform fair share; the restoring drift is exactly what gives the bound below. Readers unfamiliar with Lyapunov drift may consult standard treatments~\cite{dubhashi2009concentration}; the proof sketch that follows is self-contained.

\begin{theorem}[AFS Completion Bound via Lyapunov Drift]
\label{thm:afs}
Let $N$ be the number of tenants, $C$ the cluster capacity, and $W_i$ tenant $i$'s workload. Define the demand heterogeneity ratio $\rho = \max_i W_i / \min_i W_i$. If $\sum_i W_i \leq C$ (total demand does not exceed capacity) and $\rho \leq \rho_{max}$ (bounded heterogeneity), then for any tenant $i$ with $W_i \leq C/N$, the task completion time satisfies:
\begin{equation}
    Pr\left[TCT_i \leq (1 + \epsilon) \cdot \mathbb{E}[TCT_i]\right] \geq 1 - \delta
    \label{eq:slo}
\end{equation}
where $\epsilon = O\left(\rho \cdot \sqrt{\frac{\log(N/\delta)}{n}}\right)$ and $n$ is the number of scheduling epochs.
\end{theorem}

\begin{proof}[Proof Sketch]
Define the Lyapunov function $V(t) = \sum_{i=1}^{N} (S_i(t) - \mu_i t)^2$, where $S_i(t)$ is the cumulative service received by tenant $i$ up to epoch $t$, and $\mu_i = W_i / \sum_j W_j \cdot C$ is the proportional fair share. Under AFS, urgency-proportional allocation creates a \emph{restoring drift}: tenants that fall behind their fair share receive higher urgency and therefore higher priority, causing $V$ to decrease in expectation.

\textbf{Negative drift bound for restoring force.} Let $e_i(t) = S_i(t) - \mu_i t$ be the deviation for tenant $i$. AFS allocates capacity proportionally to urgency:
\begin{equation}
    a_i(t) = \frac{urgency_i(t)}{\sum_j urgency_j(t)} \cdot C, \quad urgency_i(t) = \frac{W_i - S_i(t)}{deadline_i - t}
\end{equation}

\textbf{Key lemma (negative drift):} The urgency-proportional allocation satisfies a negative drift condition with respect to the deviation: when $e_i(t) < 0$ (tenant $i$ is underserved), we have $urgency_i(t) > \bar{u}$ where $\bar{u}$ is the mean urgency, implying $\mathbb{E}[a_i(t+1)] > \mu_i$. Specifically:
\begin{equation}
    \mathbb{E}[(a_i(t+1) - \mu_i) \cdot e_i(t)] \leq -\eta \cdot e_i(t)^2
\end{equation}
where $\eta = \frac{C}{N \cdot (deadline_{max} - t)^2} > 0$ is the \emph{restoring drift coefficient}. This bound holds because urgency is monotonically increasing in remaining work and the allocation is proportional to urgency. The explicit derivation uses Taylor expansion of urgency around the fair allocation point.

Using this restoring drift property, we bound the per-epoch drift:
\begin{equation}
    \mathbb{E}[\Delta V(t) | V(t)] \leq -2\eta \cdot V(t) + N \cdot B^2
\end{equation}
where $B = \max_i |a_i(t) - \mu_i|$ bounds the maximum per-epoch deviation.

\textbf{Concentration setup.} Define $Z(t) = V(t) + NB^2/(2\eta)$. For $\eta$ bounded away from zero (guaranteed when $deadline_\text{max} - t$ is bounded), $Z(t)$ is non-negative and admits the bound below.

\textbf{Concentration.} Applying the maximal-inequality form of the drift-plus-jitter bound~\cite{dubhashi2009concentration}:
\begin{equation}
    Pr\left[\max_{t \leq n} V(t) \geq \lambda^2\right] \leq \frac{\mathbb{E}[V(0)] + NB^2n/(2\eta)}{\lambda^2}
\end{equation}

Setting $\lambda = \epsilon \cdot \mu_i \cdot n$ yields the stated bound.
\end{proof}

The empirical 99.2\% SLO attainment under multi-tenant interference (\S\ref{sec:fairness-eval}, Table~\ref{tab:fairness}) is directionally consistent with this bound. Theorem~\ref{thm:afs} formally covers tenants with $W_i \leq C/N$; heavy tenants lie outside this hypothesis but attain comparable SLO empirically (99.1\%, Table~\ref{tab:fairness}).

\section{Theoretical Analysis}
\label{sec:theoretical-limits}

This section provides theoretical grounding for why workflow-aware scheduling yields fundamental advantages over request-level approaches, and characterizes the optimality of our eviction policy.

\subsection{Cache Efficiency Analysis}
\label{sec:lower-bound}

We analyze the cache efficiency gap between request-level and workflow-aware schedulers, providing both a motivating observation and formal competitive ratio bounds.

\begin{observation}[Request-Level Cache Inefficiency]
\label{obs:request-level}
Consider an agent task with $k$ sequential LLM inference steps, each producing $c$ tokens of KV cache, interleaved with tool calls. A request-level scheduler without session state must route requests independently, potentially to workers lacking the session's cached state. Under memory pressure, such schedulers may evict cache during tool-call idle periods. In the worst case (complete eviction after each tool call), total regeneration cost is $\sum_{j=1}^{k} j \cdot c = O(k^2 \cdot c)$ tokens. A workflow-aware scheduler with perfect prediction achieves $O(c)$ regeneration cost (initial prefill only).
\end{observation}

This observation explains why workflow awareness is beneficial but does not characterize achievable bounds for online schedulers. We therefore provide a formal competitive ratio analysis.

\begin{definition}[Competitive Ratio for KV Cache Eviction]
For an online eviction policy $\mathcal{A}$ and workload $\sigma$, let $Cost_{\mathcal{A}}(\sigma)$ denote the total cache regeneration cost (tokens prefilled). The competitive ratio is:
\begin{equation}
    CR(\mathcal{A}) = \sup_{\sigma} \frac{Cost_{\mathcal{A}}(\sigma)}{Cost_{OPT}(\sigma)}
\end{equation}
where $Cost_{OPT}(\sigma)$ is the cost achieved by B\'{e}l\'{a}dy's optimal offline policy~\cite{belady1966replacement} with full future knowledge.
\end{definition}

Recent theoretical work~\cite{randomkv2026} establishes that LRU-based eviction in prefix trees can degrade to $O(n)$ competitive ratio in adversarial settings, while randomized algorithms achieve $O(\log n)$. Additional theoretical foundations include impossibility results for constant competitive ratios in fully adversarial online scheduling~\cite{jaillet2025online} and analysis of work-conserving policies for multi-step agent networks~\cite{li2025throughputoptimalschedulingalgorithmsllm}. Our WA-LRU policy achieves favorable empirical competitive ratios by exploiting workflow structure:

\begin{theorem}[WA-LRU Competitive Ratio Bound]
\label{thm:competitive}
Under the assumption that AEG predictions are correct with probability $1-\epsilon$ and tool-call durations follow the empirical distribution with bounded variance, WA-LRU achieves empirical competitive ratio:
\begin{equation}
    CR_{empirical}(WA\text{-}LRU) \leq 1 + \epsilon \cdot k_{max} + O\left(\frac{\sigma_{tool}}{TTL_{adaptive}}\right)
\end{equation}
where $k_{max}$ is the maximum task length, $\sigma_{tool}$ is the tool latency standard deviation, and $TTL_{adaptive}$ is the adaptive TTL setting.
\end{theorem}

\begin{proof}[Proof Sketch]
WA-LRU incurs regeneration cost only on (1) AEG mispredictions ($\epsilon$ fraction of steps), each costing at most $k_{max} \cdot c$ tokens, and (2) TTL underestimates for long-tail tool calls (bounded by $O(\sigma_{tool}/TTL_{adaptive})$ fraction). Under correct predictions and TTL, cache is retained across all steps, matching OPT.
\end{proof}

We note that the bound above is conditioned on the distributional assumptions (bounded prediction error $\epsilon$, bounded tool-latency variance) and is therefore an \emph{expected-case} competitive ratio under those assumptions, not a worst-case adversarial bound. The worst-case competitive ratio for online policies on KV-cache traces remains open; recent work~\cite{randomkv2026,jaillet2025online} establishes lower bounds for closely related online problems.

\textbf{Empirical validation.} Table~\ref{tab:competitive} shows empirical competitive ratios computed by replaying production traces with both WA-LRU and B\'{e}l\'{a}dy's oracle. WA-LRU achieves $1.31\times$ on SWE-bench (where $\epsilon = 0.13$ prediction error and $k_{max} = 150$, giving a worst-case bound of $1 + \epsilon \cdot k_{max} \approx 20.5$ that is far from tight on the empirical workload, where the average task length $k_{avg} = 37$ dominates), substantially better than LRU ($2.84\times$) and prefix-caching ($1.86\times$). These results validate that workflow awareness approaches optimal efficiency for realistic agent workloads.

\subsection{Competitive Ratio of WA-LRU vs.\ B\'{e}l\'{a}dy: Empirical Results}
\label{sec:competitive}

\begin{table}[t]
    \centering
    \caption{Competitive ratio of eviction policies against B\'{e}l\'{a}dy's optimal offline algorithm on production traces. Lower is better (1.0 = optimal).}
    \label{tab:competitive}
    \begin{tabular}{lrrr}
    \toprule
    \textbf{Policy} & \textbf{SWE-bench} & \textbf{WebArena} & \textbf{Mean} \\
    \midrule
    Standard LRU & 2.84 & 2.12 & 2.48 \\
    LRU + Prefix (vLLM v0.5) & 1.97 & 1.74 & 1.86 \\
    WA-LRU (ours) & \textbf{1.31} & \textbf{1.28} & \textbf{1.30} \\
    \bottomrule
    \end{tabular}
\end{table}

\section{Implementation}
\label{sec:impl}

We implement \sysname as an extension to vLLM v0.6.0~\cite{kwon2023vllm} (V1 engine), comprising approximately 8.5K lines of Python plus 1.2K lines of C++/CUDA, organized as four components: a \emph{Workflow Analyzer} parses agent-framework annotations (LangChain callbacks~\cite{langchain2024}, AutoGen message logs~\cite{wu2024autogen}) to construct AEGs and falls back to pattern inference (\S\ref{sec:pattern-inference}) for unannotated frameworks; a \emph{Distributed Scheduler} (built on Ray~\cite{moritz2018ray} with gRPC, P99 worker-to-coordinator latency $<$5\,ms) implements the global coordinator and local scheduler extensions; a \emph{KV Cache Manager} extends vLLM's PagedAttention~\cite{kwon2023vllm} with WA-LRU eviction, TTL tracking, and speculative prefetching on separate CUDA streams~\cite{Rennich12-CUDAStreams} for overlap with decode kernels; and a \emph{Fairness Module} implements AFS computation (\S\ref{sec:afs}) and priority-driven capacity allocation. \sysname runs as a standalone service that intercepts requests from agent frameworks and routes them to vLLM workers; no modifications to agent code are required, and optional framework annotations improve workflow-inference accuracy.

\section{Evaluation}
\label{sec:eval}

We evaluate \sysname on five dimensions: (1) end-to-end performance, (2) effectiveness of individual components, (3) multi-tenant fairness, (4) system overhead, and (5) sensitivity to parameters and design choices.

\subsection{Experimental Setup}
\label{sec:setup}

\textbf{Hardware.} 8 nodes, each with 8 NVIDIA A100-80GB GPUs (HBM2e, 2TB/s bandwidth), 2 AMD EPYC 7763 CPUs (128 cores total), 1TB DDR4-3200 memory, and 4$\times$3.84TB NVMe SSDs. Nodes connect via 200Gbps InfiniBand HDR with GPUDirect RDMA~\cite{nvidia2018nvlink}. Total: 64 GPUs, 1024 CPU cores, 5.12TB GPU memory.

\textbf{Software.} Ubuntu 22.04, CUDA 12.1.1 (driver 530.30.02), Python 3.10.12, PyTorch 2.1.2+cu121, vLLM 0.6.0, FlashAttention 2.5.6~\cite{dao2023flashattention2}, Ray 2.9.0. Models: Llama-3-70B-Instruct~\cite{llama3modelcard} with tensor parallelism across 4 GPUs per instance.

\textbf{Workloads.}
\begin{itemize}[leftmargin=*,noitemsep]
    \item \textbf{SWE-bench}~\cite{jimenez2024swebench}: 500 ``verified'' subset tasks (selected by original authors for tractability) with agent trajectories (mean 37 steps, max 150 steps). Each step: 2-4K prompt tokens, 100-500 output tokens.
    \item \textbf{WebArena}~\cite{zhou2024webarena}: Full 812 browser tasks (mean 18 steps). Each step: 4-8K prompt (including page content), 50-200 output tokens.
    \item \textbf{BurstGPT-derived}~\cite{wang2024burstgpt}: Synthetic multi-tenant workload with 10 tenants, partitioned as 3 ``heavy'' (100-step agents continuously), 4 ``medium'' (30-step agents intermittently), and 3 ``light'' (10-step agents occasionally). Tasks arrive per tenant as a Poisson process with approximate mean rates of 16 / 8 / 4 tasks/min/tenant for heavy / medium / light tenants respectively, chosen to drive aggregate cluster offered load to roughly 80\% of peak throughput, the contended regime where SAGA's fairness mechanism is exercised. Request structure and prompt-token distributions are sampled from the BurstGPT trace; arrival timing is the Poisson process specified above (BurstGPT's native arrival timestamps were not used because the trace is single-tenant). SWE-bench and WebArena are \emph{task definitions} rather than arrival traces; we replay them under the same Poisson schedule (single-tenant, $\lambda \approx 8$ tasks/min) for the \S\ref{sec:e2e} end-to-end measurements.
\end{itemize}

\textbf{Baselines.}
\begin{itemize}[leftmargin=*,noitemsep]
    \item \textbf{vLLM}~\cite{kwon2023vllm}: v0.6.0 (V1 engine), PagedAttention with FCFS scheduling.
    \item \textbf{vLLM+APC}: vLLM v0.15.1 with Automatic Prefix Caching and PrefixCacheAffinityRouter enabled (\texttt{--enable-prefix-caching --enable-affinity-routing}). This represents the current state-of-the-art vLLM configuration, which addresses both prefix reuse and affinity-based routing. Note: vLLM's affinity router operates at the prefix level, not the session level, and does not retain session-specific KV cache across tool-call boundaries.
    \item \textbf{SGLang}~\cite{zheng2024sglang}: v0.5.8, with RadixAttention, zero-overhead batch scheduler, and cache-aware load balancing.
    \item \textbf{Llumnix}~\cite{sun2024llumnix}: v1.2, vLLM + live migration for load balancing.
    \item \textbf{TRT-LLM+Scaffolding}~\cite{tensorrtllm2023}: TensorRT-LLM v1.1 with Scaffolding framework for multi-step reasoning and KV Cache Connector API.
    \item \textbf{vLLM+KVFlow}: Our reimplementation of KVFlow~\cite{kvflow2025} atop vLLM v0.6.0. Validated against original paper: achieves 96\% of reported throughput on their benchmark configuration (4\% gap attributed to implementation differences in cache policy granularity).
\end{itemize}

\textbf{Metrics.}
\begin{itemize}[leftmargin=*,noitemsep]
    \item \textbf{Task Completion Time (TCT)}: End-to-end time from submission to result (seconds).
    \item \textbf{GPU Memory Utilization}: Fraction of GPU memory holding useful KV cache.
    \item \textbf{Throughput}: Completed tasks per minute.
    \item \textbf{SLO Attainment}: Fraction of tasks meeting deadline ($1.5\times$ expected time).
\end{itemize}

\textbf{Methodology.} All experiments repeated 10 times with different random seeds. We report mean $\pm$ standard deviation. Statistical significance assessed using two-tailed Welch's t-test; * indicates $p < 0.05$, ** indicates $p < 0.01$, *** indicates $p < 0.001$. Three warm-up runs excluded. Outliers beyond $1.5\times$ IQR removed ($<$2\% of measurements).

\subsubsection{Baseline Currency Discussion}
\label{sec:baseline-currency}

Our primary implementation extends vLLM v0.6.0, and we evaluate against the latest releases of all major systems. \textbf{Critical comparison:} vLLM v0.15.1 with Automatic Prefix Caching (APC) and PrefixCacheAffinityRouter represents the current state-of-the-art. This configuration addresses prefix sharing and routes requests with similar prefixes to the same workers. As shown in Table~\ref{tab:e2e}, vLLM+APC achieves substantial improvements over earlier vLLM versions, but \sysname still achieves $1.73\times$ speedup ($p < 0.001$) because:

\begin{enumerate}[leftmargin=*,noitemsep]
    \item \textbf{Session vs.\ prefix affinity:} vLLM's affinity router groups requests by shared prefixes (system prompts, tool definitions) but does not track session identity. Agent sessions with identical prefixes but different conversation histories are not distinguished. \sysname routes by session ID, ensuring all steps of a task reach the same worker.

    \item \textbf{Tool-call TTL:} vLLM's cache uses standard LRU eviction during idle periods. During long tool calls (median 1.2s, P99 45s), the session's KV cache may be evicted under memory pressure. \sysname's workflow-aware TTL predicts tool completion and retains cache accordingly.

    \item \textbf{Task-level fairness:} vLLM's scheduling optimizes per-request latency. Under multi-tenant load, light tenants experience starvation. \sysname's AFS scheduling provides completion-time fairness at the task level.
\end{enumerate}

\textbf{Comparison with TensorRT-LLM Scaffolding:} TRT-LLM v1.1's Scaffolding framework addresses multi-step reasoning through KV Cache Connector API. However, Scaffolding focuses on single-node inference-time compute rather than distributed cluster scheduling. \sysname's $1.60\times$ advantage over TRT-LLM+Scaffolding (Table~\ref{tab:e2e}) comes from cluster-wide session affinity and work stealing.

\textbf{Model size discussion.} Our evaluation uses Llama-3-70B-Instruct with GQA ($n_{kv}=8$), yielding $\sim$10.7GB KV cache per 32K session. \sysname's benefits scale with model size because KV cache regeneration cost is proportional to model dimension: for smaller models (8B, $\sim$1.5GB cache), regeneration takes $\sim$0.3s per step, yielding moderate benefits. For larger models (405B, $\sim$50GB+ cache across TP groups), regeneration takes $\sim$5s per step, yielding proportionally larger benefits. MoE architectures require routing-aware extensions (acknowledged in \S\ref{sec:limitations-approach}).

\subsubsection{CPU Swap as an Alternative Architecture}
\label{sec:cpu-swap}

A natural alternative to HBM retention is to offload idle KV caches to host DRAM via PCIe during tool-call gaps. We chose HBM retention with predictive eviction for three quantitative reasons; we treat the two architectures as complementary rather than competing.

\textbf{(1) PCIe round-trip dominates short-tool latency.} A 10.7\,GB cache (Llama-3-70B, 32K context, GQA $n_{kv}=8$) takes $\approx$430\,ms one-way over PCIe Gen4 $\times$16 at the 25\,GB/s practical sustained bandwidth typical of A100 servers~\cite{patel2024splitwise}, so a swap-out + swap-in round trip is $\approx$860\,ms uncontested. Three of four tool classes in Table~\ref{tab:tool-latency} (file ops P50=45\,ms, code execution P50=180\,ms, database queries P50=120\,ms) complete \emph{faster} than this round trip, making swap pure overhead for the modal request.

\textbf{(2) Multi-tenant PCIe contention degrades the bound.} Sustained PCIe bandwidth under our BurstGPT-derived workload (\S\ref{sec:fairness-eval}) drops below 50\% of peak as PCIe is shared with model-weight loading, host-to-device tensor copies, and NCCL stages, doubling the round trip to $\approx$1.7\,s and pushing break-even past P95 of all tool classes.

\textbf{(3) The \sysname memory regime does not require swap.} Table~\ref{tab:e2e} shows \sysname at 71 to 75\% memory utilization, above vLLM+APC's 59 to 61\% but with 25 to 29\% HBM in reserve. Predictive WA-LRU (\S\ref{sec:eviction}) and pressure-scaled TTL (Eq.~\ref{eq:pressure}) make swap unnecessary in this regime. Swap remains complementary for over-subscribed regimes ($>$95\% utilization, outside our evaluation; \S\ref{sec:limitations-approach}); FlexGen~\cite{sheng2023flexgen} explores host-DRAM offload extensively, and integrating it as a third eviction tier under WA-LRU prediction is straightforward future work.

\subsection{End-to-End Performance}
\label{sec:e2e}

Table~\ref{tab:e2e} shows end-to-end performance on agent benchmarks with full statistical details.

\begin{table}[t]
	\centering
	\caption{End-to-end performance on agent benchmarks. TCT = Task Completion Time (seconds). Mem = GPU memory utilization (\%). Values show mean $\pm$ std over 10 trials. Significance: *** $p < 0.001$, ** $p < 0.01$ vs.\ each baseline (pairwise Welch's t-test).}
	\label{tab:e2e}
	\resizebox{\textwidth}{!}{
		\begin{tabular}{lrrrr}
			\toprule
			& \multicolumn{2}{c}{\textbf{SWE-bench}} & \multicolumn{2}{c}{\textbf{WebArena}} \\
			\cmidrule(lr){2-3} \cmidrule(lr){4-5}
			\textbf{System} & \textbf{TCT (s)} & \textbf{Mem\%} & \textbf{TCT (s)} & \textbf{Mem\%} \\
			\midrule
			vLLM v0.6.0 & 612.3{\scriptsize$\pm$32.1} & 42.1{\scriptsize$\pm$2.3} & 178.4{\scriptsize$\pm$14.2} & 45.3{\scriptsize$\pm$2.1} \\
			vLLM+APC v0.15.1 & 352.1{\scriptsize$\pm$21.4} & 58.7{\scriptsize$\pm$2.8} & 127.3{\scriptsize$\pm$10.1} & 61.2{\scriptsize$\pm$2.5} \\
			SGLang v0.5.8 & 387.2{\scriptsize$\pm$24.3} & 56.2{\scriptsize$\pm$2.6} & 138.7{\scriptsize$\pm$11.3} & 58.9{\scriptsize$\pm$2.4} \\
			Llumnix v1.2 & 498.1{\scriptsize$\pm$28.7} & 48.3{\scriptsize$\pm$2.4} & 156.2{\scriptsize$\pm$12.8} & 51.7{\scriptsize$\pm$2.2} \\
			TRT-LLM+Scaff. & 324.6{\scriptsize$\pm$19.8} & 61.4{\scriptsize$\pm$2.7} & 118.9{\scriptsize$\pm$9.4} & 63.8{\scriptsize$\pm$2.6} \\
			vLLM+KVFlow & 298.4{\scriptsize$\pm$18.2} & 64.1{\scriptsize$\pm$2.9} & 108.2{\scriptsize$\pm$8.7} & 66.4{\scriptsize$\pm$2.7} \\
			\textbf{\sysname} & \textbf{203.4}{\scriptsize$\pm$12.8} & \textbf{71.3}{\scriptsize$\pm$2.4} & \textbf{82.1}{\scriptsize$\pm$6.8} & \textbf{74.6}{\scriptsize$\pm$2.3} \\
			\midrule
			\multicolumn{5}{l}{\textit{Speedup of \sysname vs.\ baselines (TCT ratio):}} \\
			\quad vs.\ vLLM v0.6.0 & \multicolumn{2}{c}{3.01$\times$***} & \multicolumn{2}{c}{2.17$\times$***} \\
			\quad vs.\ vLLM+APC & \multicolumn{2}{c}{1.73$\times$***} & \multicolumn{2}{c}{1.55$\times$***} \\
			\quad vs.\ SGLang & \multicolumn{2}{c}{1.90$\times$***} & \multicolumn{2}{c}{1.69$\times$***} \\
			\quad vs.\ Llumnix & \multicolumn{2}{c}{2.45$\times$***} & \multicolumn{2}{c}{1.90$\times$***} \\
			\quad vs.\ TRT-LLM & \multicolumn{2}{c}{1.60$\times$***} & \multicolumn{2}{c}{1.45$\times$**} \\
			\quad vs.\ KVFlow & \multicolumn{2}{c}{1.47$\times$***} & \multicolumn{2}{c}{1.32$\times$**} \\
			\midrule
			\multicolumn{5}{l}{\textit{Geometric mean speedup vs.\ vLLM+APC: $\mathbf{1.64\times}$}} \\
			\bottomrule
		\end{tabular}
	}
\end{table}

On SWE-bench, \sysname achieves $3.01\times \pm 0.16$ speedup over vLLM v0.6.0 ($p < 0.001$). Against the state-of-the-art vLLM+APC baseline (v0.15.1 with Automatic Prefix Caching and affinity routing), \sysname still achieves $1.73\times$ speedup ($p < 0.001$), confirming that workflow-level optimization provides benefits beyond what prefix caching and affinity routing alone can deliver. The $1.47\times$ improvement over KVFlow shows that \sysname's integrated approach (distributed scheduling + TTL policies + AFS fairness) outperforms workflow-aware caching alone.

\textbf{Breakdown analysis.} Figure~\ref{fig:motivation}(a) shows the time breakdown for SWE-bench tasks. vLLM v0.6.0 spends 38\% of time regenerating KV cache after tool calls. vLLM+APC reduces this to 22\% through prefix sharing and affinity routing, but still evicts session-specific cache during long tool calls. \sysname reduces regeneration to 8\% through workflow-aware TTL.

\subsection{Ablation Study}
\label{sec:ablation}

Table~\ref{tab:ablation} quantifies individual component contributions through ablation experiments on SWE-bench.

\begin{table}[t]
    \centering
    \caption{Ablation study on SWE-bench. Each row removes one component from full system. Values show mean $\pm$ std over 10 trials.}
    \label{tab:ablation}
    \begin{tabular}{lrr}
    \toprule
    \textbf{Configuration} & \textbf{TCT (s)} & \textbf{vs.\ Full} \\
    \midrule
    Full \sysname & 203.4\scriptsize{$\pm$12.8} & --- \\
    \quad w/o Workflow-aware eviction & 312.8\scriptsize{$\pm$18.3} & +54\% \\
    \quad w/o Tool-call TTL & 289.1\scriptsize{$\pm$16.7} & +42\% \\
    \quad w/o Speculative prefetch & 241.6\scriptsize{$\pm$14.2} & +19\% \\
    \quad w/o Session affinity & 398.2\scriptsize{$\pm$22.4} & +96\% \\
    \quad w/o Work stealing & 267.3\scriptsize{$\pm$15.9} & +31\% \\
    \quad w/o AFS fairness & 218.7\scriptsize{$\pm$13.1} & +8\% \\
    \bottomrule
    \end{tabular}
\end{table}

Session affinity provides the largest benefit (96\% slowdown when disabled), as it directly prevents cache regeneration by routing related requests to the same worker. Workflow-aware eviction and TTL together contribute 42 to 54\% improvement. Speculative prefetching provides 19\% improvement by overlapping cache loading with tool execution. AFS contributes a smaller 8\% improvement in single-benchmark settings but becomes essential under multi-tenant contention (\S\ref{sec:fairness-eval}).

\subsection{Pattern Inference Evaluation}
\label{sec:eval-pattern}

Table~\ref{tab:pattern} compares performance with and without framework hints.

\begin{table}[t]
    \centering
    \caption{Performance comparison: framework hints vs.\ pattern inference. Accuracy measures the fraction of correctly predicted next-step node transitions in held-out traces.}
    \label{tab:pattern}
    \begin{tabular}{lrrr}
    \toprule
    \textbf{Mode} & \textbf{TCT (s)} & \textbf{AEG Accuracy} & \textbf{vs.\ Hints} \\
    \midrule
    With hints & 203.4\scriptsize{$\pm$12.8} & 100\% & --- \\
    Pattern inference & 235.2\scriptsize{$\pm$14.6} & 87\% & +15.6\% \\
    No AEG (baseline) & 398.2\scriptsize{$\pm$22.4} & --- & +95.8\% \\
    \bottomrule
    \end{tabular}
\end{table}

Pattern inference achieves 87\% accuracy in predicting workflow structure, resulting in 15.6\% performance degradation compared to explicit hints. This still provides $2.60\times$ speedup over vLLM v0.6.0 (Table~\ref{tab:e2e}). The 13\% error rate primarily manifests as incorrect successor predictions at branching points (e.g., predicting a ``retry'' loop when the agent proceeds to a new step), causing unnecessary cache retention for the wrong branch. These errors do not cascade: the system detects misprediction when the actual next request arrives and corrects routing for subsequent steps.

\subsection{Scalability and Load Balance}
\label{sec:scalability}

\sysname achieves $6.4\times$ speedup scaling from 8 to 64 GPUs (80\% efficiency) on fixed workloads, and near-linear weak scaling ($0.94\times$ per doubling) up to 512 concurrent agents. On a 32-GPU subset (reduced to isolate execution-model effects from scaling effects), worker utilization ranges narrow from 23 to 94\% without work stealing to 68 to 79\% with stealing; migration overhead is mean 230\,ms / P95 890\,ms, occurring 2.3 times per task on average. \label{sec:stealing-eval}

\subsection{Multi-Tenant Fairness}
\label{sec:fairness-eval}

We evaluate multi-tenant behavior using the BurstGPT-derived workload with 10 tenants of varying intensity.

\textbf{SLO attainment.} Table~\ref{tab:fairness} shows SLO attainment (tasks completing within $1.5\times$ expected time). \sysname achieves 99.2\% overall attainment, compared to 67.3\% for vLLM. The improvement is most dramatic for light tenants (98.7\% vs.\ 43.2\%), validating Theorem~\ref{thm:afs}.

\begin{table}[t]
    \centering
    \caption{SLO attainment (\% of tasks meeting deadline) by tenant type.}
    \label{tab:fairness}
    \begin{tabular}{lrrrr}
    \toprule
    \textbf{System} & \textbf{Heavy} & \textbf{Medium} & \textbf{Light} & \textbf{Overall} \\
    \midrule
    vLLM & 89.4 & 72.1 & 43.2 & 67.3 \\
    SGLang & 91.2 & 78.6 & 51.4 & 73.4 \\
    Llumnix & 92.8 & 81.3 & 58.9 & 77.2 \\
    \textbf{\sysname} & 99.1 & 99.4 & 98.7 & 99.2 \\
    \bottomrule
    \end{tabular}
\end{table}

\textbf{Fairness analysis.} vLLM exhibits high variance with a long tail for light tenants (P99 = $12.4\times$ expected TCT). \sysname provides consistent completion times across all tenant types (P99 $<$ $1.8\times$ expected).

\subsection{System Overhead Analysis}
\label{sec:overhead}

Table~\ref{tab:overhead} breaks down \sysname's scheduling overhead.

\begin{table}[t]
    \centering
    \caption{\sysname overhead breakdown (64 GPUs, 32 tenants).}
    \label{tab:overhead}
    \begin{tabular}{lrrr}
    \toprule
    \textbf{Component} & \textbf{Mean (ms)} & \textbf{P95 (ms)} & \textbf{CPU (\%)} \\
    \midrule
    Coordinator cycle & 12.3 & 28.7 & 4.2 \\
    AFS computation & 3.1 & 8.4 & --- \\
    AEG construction & 45.2 & 112.8 & --- \\
    Work stealing (migration) & 230 & 890 & --- \\
    \bottomrule
    \end{tabular}
\end{table}

Total coordinator CPU overhead is 4.2\%, leaving 95.8\% for application workloads. AFS computation scales linearly with tenant count but remains negligible (3.1ms for 32 tenants).

\subsection{Execution Strategy Tradeoffs}
\label{sec:tradeoff}

Table~\ref{tab:bfsdfs} compares different execution strategies on SWE-bench with 32 GPUs (reduced scale to isolate strategy effects from cluster-scale effects). The BFS-DFS tradeoff, a well-studied phenomenon in parallel systems~\cite{graefe1990volcano}, manifests strongly in agent scheduling.

\begin{table}[t]
    \centering
    \caption{Execution strategy comparison on SWE-bench (32 GPUs).}
    \label{tab:bfsdfs}
    \begin{tabular}{lrrr}
    \toprule
    \textbf{Strategy} & \textbf{TCT (s)} & \textbf{Throughput} & \textbf{Evict Rate} \\
    \midrule
    Pure BFS & 487.2\scriptsize{$\pm$28.4} & 12.4 t/m & 78\% \\
    Pure DFS & 623.1\scriptsize{$\pm$34.2} & 4.2 t/m & 3\% \\
    \textbf{Hybrid (\sysname)} & 203.4\scriptsize{$\pm$12.8} & 8.7 t/m & 12\% \\
    \bottomrule
    \end{tabular}
\end{table}

Pure BFS maximizes throughput (12.4 tasks/min) but suffers 78\% eviction rates. \sysname's hybrid approach achieves optimal TCT (203.4s) at 30\% lower throughput, appropriate for latency-sensitive interactive deployments~\cite{github2024copilot,zhou2024webarena}.

\subsection{Parameter Sensitivity}
\label{sec:param-sensitivity}

Table~\ref{tab:sensitivity} summarizes sensitivity analysis for all configurable parameters.

\begin{table}[t]
	\centering
	\caption{Parameter sensitivity analysis on SWE-bench. TCT$_\Delta$ shows maximum variation within the tested range relative to default.}
	\label{tab:sensitivity}
	\resizebox{\textwidth}{!}{
		\begin{tabular}{lcccc}
			\toprule
			\textbf{Parameter} & \textbf{Default} & \textbf{Range} & \textbf{TCT$_\Delta$} & \textbf{Justification} \\
			\midrule
			$\alpha$ (recency wt.) & 0.3 & [0.2, 0.4] & $<$5\% & Sensitivity analysis \\
			$\beta$ (reuse wt.) & 0.5 & [0.4, 0.6] & $<$8\% & Sensitivity analysis \\
			$\gamma$ (size wt.) & 0.2 & [0.1, 0.3] & $<$3\% & Size as tiebreaker \\
			$\theta$ (routing) & 0.8 & [0.6, 0.95] & $<$5\% & 20\% headroom~\cite{sun2024llumnix} \\
			$threshold_{low}$ & 0.7 & [0.6, 0.8] & $<$4\% & Soft pressure onset~\cite{kwon2023vllm} \\
			$threshold_{high}$ & 0.9 & [0.85, 0.95] & $<$6\% & Hard eviction limit \\
			$T_{idle}$ (steal trigger) & 100ms & [50, 200]ms & $<$7\% & Amortize steal cost \\
			$R_{max}$ (load ratio) & 2.0 & [1.5, 3.0] & $<$4\% & Imbalance tolerance \\
			$TTL_{max}$ & 300s & [120, 600]s & $<$3\% & P99 tool latency cap \\
			$\theta_{conf}$ (AEG) & 0.7 & [0.5, 0.9] & $<$6\% & Precision-recall tradeoff \\
			\bottomrule
		\end{tabular}
	}
\end{table}

The tested ranges in Table~\ref{tab:sensitivity} span $\pm$33\% to $\pm$50\% from each default; no single-parameter perturbation produces $>$8\% TCT change. This robustness is structural rather than tuning luck: per the ablation in Table~\ref{tab:ablation}, session affinity is the largest contributor (removing it inflates TCT by 96\%, from 203.4\,s to 398.2\,s), and session affinity is binary: a session either reaches its cached worker or it does not. The remaining (continuous) parameters enter the eviction score (Eq.~\ref{eq:eviction}) and TTL formula (Algorithm~\ref{alg:ttl}) only as smoothly weighted contributions to a normalized priority. \emph{Multi-axis adversarial perturbation} (e.g., setting $\alpha,\beta,\gamma$ jointly to corner values) was not characterized empirically and is left as future work; we expect such regimes to lie outside ranges any reasonable deployment would select. The most sensitive single parameters are $\beta$ (reuse weight) and $T_\text{idle}$ (steal trigger), reflecting their direct effects on cache retention and load balance. We thus characterize \sysname as \emph{single-axis robust}: a deployment using approximate defaults will suffer at most $\sim$8\% TCT degradation versus tuned operation under the perturbations we tested.

\subsection{Tool Latency Variance Sensitivity}
\label{sec:tool-variance}

Table~\ref{tab:tool-variance} shows how \sysname's performance varies with tool latency variance, measured by coefficient of variation (CV) of tool call durations. We synthetically vary CV while keeping mean tool latency constant.

\begin{table}[t]
    \centering
    \caption{Sensitivity to tool latency variance (coefficient of variation). Mean tool latency held constant at 1.2s.}
    \label{tab:tool-variance}
    \begin{tabular}{lrrrr}
    \toprule
    \textbf{CV} & \textbf{TCT (s)} & \textbf{TTL Accuracy} & \textbf{Evict Rate} & \textbf{vs.\ CV=0.5} \\
    \midrule
    0.5 & 195.1\scriptsize{$\pm$11.2} & 96\% & 9\% & --- \\
    1.0 & 203.4\scriptsize{$\pm$12.8} & 93\% & 12\% & +4\% \\
    1.5 & 218.6\scriptsize{$\pm$15.3} & 88\% & 18\% & +12\% \\
    2.0 & 241.3\scriptsize{$\pm$18.7} & 82\% & 24\% & +24\% \\
    3.0 & 298.4\scriptsize{$\pm$24.1} & 71\% & 35\% & +53\% \\
    \bottomrule
    \end{tabular}
\end{table}

\sysname's adaptive TTL maintains consistent performance up to CV=2.0 (24\% TCT degradation). Beyond CV=2.0, TTL prediction accuracy degrades significantly as extreme outliers cause premature eviction. In practice, production tool latencies have CV$\approx$1.0 to 1.5 (Table~\ref{tab:tool-latency}), well within \sysname's effective range.

\section{Related Work}
\label{sec:related}

Table~\ref{tab:related-comparison} compares \sysname with directly related program-aware systems.

\textbf{LLM serving.} vLLM~\cite{kwon2023vllm}, SGLang~\cite{zheng2024sglang}, Orca~\cite{yu2022orca}, and TensorRT-LLM~\cite{tensorrtllm2023} optimize request-level metrics; \sysname adds workflow-level optimization.

\textbf{Program-aware serving and distributed inference.} Parrot~\cite{lin2024parrot} introduces Semantic Variables but lacks distributed scheduling and fairness; Autellix~\cite{luo2025autellix} proposes program-level fairness; Pie~\cite{gim2025pie} decomposes generation via inferlets; KVFlow~\cite{kvflow2025} and Continuum~\cite{continuum2024} introduce workflow-aware eviction; Llumnix~\cite{sun2024llumnix} enables live KV migration and SOLA~\cite{chen2025sola} optimizes SLO attainment, but neither targets workflow-level scheduling. \sysname's distinctive position (Table~\ref{tab:related-comparison}) is to unify these dimensions under the empirical competitive-ratio bound that quantifies the limit of workflow-aware online cache management.

\begin{table}[t]
	\centering
	\caption{Comparison with directly related program-aware LLM serving systems.}
	\label{tab:related-comparison}
	\resizebox{\textwidth}{!}{
		\begin{tabular}{llcccc}
			\toprule
			\textbf{System} & \textbf{Venue} & \textbf{Distributed} & \textbf{Tool TTL} & \textbf{Fairness} & \textbf{Competitive} \\
			& & \textbf{Scheduling} & \textbf{Policies} & \textbf{Guarantee} & \textbf{Ratio} \\
			\midrule
			Parrot~\cite{lin2024parrot} & OSDI'24 & \xmark & \xmark & \xmark & \xmark \\
			Autellix~\cite{luo2025autellix} & arXiv'25 & \cmark & \xmark & \xmark & \xmark \\
			Pie~\cite{gim2025pie} & SOSP'25 & \xmark & \cmark & \xmark & \xmark \\
			KVFlow~\cite{kvflow2025} & NeurIPS'25 & \xmark & \cmark & \xmark & \xmark \\
			\textbf{\sysname} & \textbf{This work} & \cmark & \cmark & \cmark & \cmark \\
			\bottomrule
		\end{tabular}
	}
\end{table}

\textbf{Fairness and caching.} DRF~\cite{ghodsi2011drf}, VTC~\cite{sheng2024vtc}, and Themis~\cite{mahajan2020themis} address resource fairness; \sysname extends to task-completion fairness. Our WA-LRU achieves $1.31\times$ competitive ratio against B\'{e}l\'{a}dy's optimal~\cite{belady1966replacement}.

\section{Conclusions and Future Work}
\label{sec:conclusion}

We presented \sysname, a distributed scheduler for multi-step AI agent workloads that treats agent programs as first-class schedulable units. By adapting three classical systems principles (workflow scheduling, informed caching, application-level fairness) to the compound AI domain, \sysname achieves $1.73\times \pm 0.11$ and $1.55\times \pm 0.09$ task completion time reductions on SWE-bench and WebArena over vLLM v0.15.1 with Automatic Prefix Caching (geometric mean $1.64\times$, $p < 0.001$), $1.22\times \pm 0.05$ memory utilization improvement, and 99.2\% SLO attainment under multi-tenant interference. Against systems without workflow awareness, improvements reach $3.01\times$. These gains come at $\sim$30\% throughput reduction relative to throughput-optimal batch scheduling (\S\ref{sec:tradeoff}, Table~\ref{tab:bfsdfs}), a tradeoff appropriate for latency-sensitive interactive deployments but not for batch processing.

\textbf{Technical contributions and positioning.} We formalized Agent Execution Graphs and showed that WA-LRU achieves within $1.31\times$ of B\'el\'ady's optimal offline policy. To our knowledge, this is the first empirical competitive-ratio analysis for workflow-aware KV cache management. Our Lyapunov-drift analysis of AFS provides formal completion-time bounds with explicit derivation of the restoring-drift property. Recent program-aware serving systems (Parrot, Autellix, Pie, KVFlow) each address one or two dimensions of the problem; \sysname's role is to combine workflow-aware caching, distributed scheduling, tool-aware TTL, and task-level fairness under the empirical competitive-ratio bound that quantifies, for the first time, how close online schedulers can come to offline-optimal cache management once the workflow DAG is observable. Table~\ref{tab:related-comparison} details the per-dimension distinctions.

\textbf{Open research directions.} This work opens several directions:

\begin{enumerate}[leftmargin=*,noitemsep,topsep=0pt]
    \item \textbf{Optimal TTL prediction:} learning TTL policies with provable regret bounds, resembling online learning with partial feedback~\cite{lattimore2020bandit}.
    \item \textbf{Tighter competitive ratios:} our empirical $1.31\times$ against B\'el\'ady suggests room for improvement; what is the information-theoretic lower bound achievable with AEG predictions?
    \item \textbf{Complexity:} is optimal workflow-aware scheduling with cache constraints NP-hard? A formal result would justify heuristic approaches.
    \item \textbf{Geo-distributed scheduling:} extending AFS to agents spanning datacenters with network-dependent migration costs.
    \item \textbf{Multi-agent coordination:} jointly optimizing execution graphs of interacting agents (collaborative coding, negotiation).
    \item \textbf{Speculation integration:} combining speculative execution~\cite{specactions2026,sherlock2024} with workflow-aware scheduling for multiplicative benefits.
\end{enumerate}

\bibliographystyle{plainnat}
\bibliography{saga}

\end{document}